\documentclass[11pt,fleqn]{article}
\usepackage{fullpage}
\usepackage{amsmath,amssymb,amsthm}
\usepackage{xspace}
\usepackage{ifpdf}
\ifpdf    
\usepackage{hyperref}
\else    
\usepackage[hypertex]{hyperref}
\fi

\newcommand\remove[1]{}

\newcommand{\veps}{\varepsilon}
\newcommand{\R}{{\mathbb R}}

\newcommand{\minn}[1]{\min\{{#1}\}}

\DeclareMathOperator{\diam}{diam}

\newcommand{\lip}{Lipschitz\xspace}
\newcommand{\lipnorm}[1]{\|#1\|_{\mathrm{Lip}}}

\newcommand{\tilh}{{\tilde h}}
\newcommand{\tilG}{{\tilde G}}
\newcommand{\ceil}[1]{\lceil{#1}\rceil}

\newcommand{\JL}{{\mathrm{JL}}}
\newcommand{\PsiJL}{\Psi_\JL}
\DeclareMathOperator{\dimC}{dim}

\newcommand{\eqdef}{\stackrel{{\rm def}}{=}}

\def\Lovasz{Lov{\'a}sz\xspace}
\def\tO{{\tilde O}}

\newtheorem{theorem}{Theorem}[section]
\newtheorem{lemma}[theorem]{Lemma}

\newtheorem{property}[theorem]{Property}
\newtheorem{question}{Question}

\numberwithin{equation}{section}

\theoremstyle{definition}

\def\compactify{\itemsep=0pt \topsep=0pt \partopsep=0pt \parsep=0pt}

\newcounter{this-list}

\newcounter{par-list}
\newlength{\parlistlength}

\title{A Nonlinear Approach to Dimension Reduction%
\thanks{A preliminary version of this paper appeared 
in Proceedings of SODA 2011.
The current version includes some previously omitted proof material 
and various minor corrections.
This work was supported in part by The Israel Science Foundation
(grant \#452/08), and by a Minerva grant.
}
}
\author{Lee-Ad Gottlieb%
\thanks{Department of Computer Science and Mathematics, Ariel University, Israel.
Email: \texttt{leead@ariel.ac.il}
}
\\Ariel University
\and
Robert Krauthgamer%
\thanks{Weizmann Institute of Science, Rehovot, Israel.
Email: \texttt{robert.krauthgamer@weizmann.ac.il}
}
\\ Weizmann Institute of Science
}

\begin{document}
\maketitle

\begin{abstract} 
The $\ell_2$ flattening lemma of Johnson and Lindenstrauss \cite{JL84} is 
a powerful tool for dimension reduction. It has been conjectured that 
the target dimension bounds can be refined and bounded
in terms of the intrinsic dimensionality of the data 
set (for example, the doubling dimension). One such problem was proposed 
by Lang and Plaut \cite{LP01} (see also 
\cite{GKL03,MatousekProblems07,ABN08,CGT10}), and is still open. We prove 
another result in this line of work:

\begin{quote}
{\em The snowflake metric $d^\alpha$ ($\alpha<1$) of a doubling set 
$S\subset\ell_2$ embeds with constant distortion
into $\ell_2^D$, for dimension $D$ that depends solely on the doubling 
constant of the metric. }
\end{quote}

In fact, the distortion can be made arbitrarily close to $1$, and the target 
dimension is polylogarithmic in the doubling constant. Our techniques are 
robust and extend to the more difficult space $\ell_1$,
although the dimension bounds here are quantitatively inferior than those 
for $\ell_2$. \end{abstract}


\section{Introduction}\label{sec:intro}

Dimension reduction, in which high-dimensional data is faithfully represented
in a low-dimensional space, is a key tool in several fields.
Probably the most prevalent mathematical formulation of this problem 
considers the data to be a set $S\subset \ell_2$,
and the goal is to map the points in $S$ into a low-dimensional $\ell_2^k$.
(Here and throughout, $\ell_p^k$ denotes the space $\R^k$ endowed with the
$\ell_p$-norm; $\ell_p$ is its generalization to countably many coordinates,
namely all sequences that are $p$-th power summable.)
A celebrated result in this area is the so-called JL-Lemma:
\begin{theorem}[Johnson and Lindenstrauss \cite{JL84}] 
\label{thm:JL}
For every $n$-point subset $S\subset \ell_2$ and every $0<\veps<1$,
there is a mapping $\PsiJL:S\to \ell_2^k$
that preserves all interpoint distances in $S$ within factor $1+\veps$,
and has target dimension $k=O(\veps^{-2}\log n)$.
\end{theorem}

This positive result is remarkably strong;
in fact the map $\PsiJL$ is an easy to describe (randomized) linear transformation.
It has found many applications, and has become a basic tool.
It is natural to seek the optimal (minimum) target dimension $k$ 
possible in this theorem.
The logarithmic dependence on $n=|S|$ is necessary,
as can be easily seen by volume arguments,
and Alon \cite{Alon03} further proved that the JL-Lemma is optimal 
up to a factor of $O(\log\tfrac1\veps)$.
These lower bounds are existential, meaning that there are sets $S$ 
for which the result of the JL-Lemma cannot be improved.
However, it may still be possible to significantly reduce the dimension 
for sets $S$ that are ``intrinsically'' low-dimensional.
This raises the interesting and fundamental question of bounding $k$ 
in terms of parameters other than $n$, which we formalize next.

We recall some basic terminology involving metric spaces.%
\remove{\footnote{A metric space is a set $M$ of points endowed with a 
distance function $d_M(\cdot,\cdot)$ that is nonnegative, symmetric, 
satisfies the triangle inequality and in addition, 
$d_M(x,y)=0$ if and only if $x=y$. }}
The {\em doubling constant} of a metric $(M,d_M)$, denoted $\lambda(M)$,
is the smallest $\lambda\ge 1$ such that every (metric) ball in $M$
can be covered by at most $\lambda$ balls of half its radius.
We say that $M$ is {\em doubling}
if its doubling constant $\lambda(M)$ is bounded independently of $|M|$.
It is sometimes more convenient to refer to $\dimC(M)\eqdef\log_2 \lambda(M)$, 
which is known as the {\em doubling dimension} of $M$ \cite{GKL03}. 
An {\em embedding} of one metric space $(M,d_M)$ into another $(N,d_N)$
is a map $\Psi:M\to N$.
We say that $\Psi$ attains {\em distortion} $D'\ge 1$
if $\Psi$ preserves every pairwise distance within factor $D'$,
namely, there is a scaling factor $s>0$ such that
$$
  1 \le \frac{d_N(\Psi(x),\Psi(y))}{s\cdot d_M(x,y)} \le D',
  \qquad \forall x,y\in M.
$$
The following problem was posed independently
by \cite{LP01} and \cite{GKL03} (see also 
\cite{MatousekProblems07,ABN08,CGT10}):

\begin{question}\label{q2}
Does every doubling subset $S\subset\ell_2$ embed with distortion $D'$
into $\ell_2^D$ for $D,D'$ that depend only on $\lambda(S)$?
\end{question}

This question is still open and seems very challenging.
Resolving it in the affirmative seems to require completely different
techniques than the JL-Lemma, since such an embedding 
cannot be achieved by a linear map \cite[Remark 4.1]{IN07}.
For algorithmic applications, it would be ideal to resolve positively
an even stronger variant of Question \ref{q2}, 
where the target distortion $D'$ is an absolute constant 
independent of $\lambda(S)$, or even $1+\veps$ as in the JL-Lemma.
This stronger version has not been excluded, and is still open as well.

\subsection{Results and Techniques}
\label{sec:results}

We present dimension reduction results for doubling subsets of Euclidean spaces.
In fact, we devise a robust framework that extends even 
to the spaces $\ell_1$ and $\ell_\infty$.
Our results incur constant or $1+\veps$ distortion, 
with target dimension that depends not on $|S|$ but rather on $\dimC(S)$
(and this dependence is unavoidable due to volume arguments).
We remark that such guarantees -- very low distortion and dimension --
are highly sought-after in metric embeddings, but rarely achieved.
We state our results in the context of finite metrics (subsets of $\ell_p$);
they extend to infinite subsets of $L_p$ 
via standard arguments.

\paragraph{Snowflake Embedding.}

Our primary embedding achieves distortion $1+\veps$ for
the {\em snowflake metric} $d^\alpha$ of an input metric $d$
(i.e.\ the snowflake metric is obtained by raising every pairwise 
distance to power $0<\alpha<1$).
It is instructive to view $\alpha$ as a fixed constant, say $\alpha=1/2$.
We prove the following in Section \ref{sec:l2betterSnowflake}.
Throughout, we use $\tilde O(f)$ to denote $f\cdot (\log f)^{O(1)}$.

\begin{theorem} \label{thm:main}
Let $0<\veps<1/4$, $0<\alpha<1$, and 
$\tilde{\alpha} = \min \{ \alpha,1-\alpha\}$.
Every finite subset $S\subset \ell_2$ admits an embedding $\Phi:S\to \ell_2^k$
for $k=\tilde O(\veps^{-4}\tilde{\alpha}^{-2}(\dimC S)^2)$, such that
$$
   1
   \le \frac{\|\Phi(x)-\Phi(y)\|_2}{\|x-y\|_2^{\alpha}}
   \le 1+\veps,
   \qquad \forall x,y\in S.
$$
\end{theorem}

Notice the difference between our theorem and Question \ref{q2}:
Our embedding achieves better distortion $1+\veps$,
but it applies to the (often easier) snowflake metric $d^\alpha$.
Our result is also related to the following theorem of Assouad \cite{Assouad83}:
{\em For every doubling metric $(M,d)$ and every $0<\alpha<1$,
the snowflake metric $d^\alpha$ embeds into $\ell_2^{D}$
with distortion $D'$, where $D,D'$ depend only on $\lambda(M)$ and $\alpha$.}
%
Note the theorem's vast generality -- the only requirements are the doubling property
(which by volume arguments is an obvious necessity) 
and that the data be a metric -- 
at the nontrivial price that the distortion achieved depends on $\lambda(M)$.
Compared to Assouad's theorem, our embedding achieves 
a much stronger distortion $1+\veps$,
but requires the additional assumption that the input metric is Euclidean.

Previously, Theorem \ref{thm:main} was only known to hold 
in the special case where $S=\R$ (the real line).
For this case, Kahane \cite{Kahane81} and Talagrand \cite{Talagrand92} 
exhibit a $1+\veps$ distortion embedding
of the snowflake metric $|x-y|^\alpha$ into $\ell_2^k$.
Kahane \cite{Kahane81} shows an embedding of $|x-y|^{1/2}$ 
(also known as Wilson's helix) into dimension $k=O(1/\veps)$,
while Talagrand \cite{Talagrand92} shows how to embed 
every snowflake metric $|x-y|^\alpha$, $\alpha\in(0,1)$,
with dimension $k=O(K(\alpha)/\veps^2)$.
Theorem \ref{thm:main} can be viewed as a generalization of \cite{Kahane81,Talagrand92}
to arbitrary doubling subset of $\ell_2$ (or other $\ell_p$),
albeit with a somewhat worse dependence on $\veps$.

\paragraph{Embedding for a Single Scale.}

Most of our technical work is devoted to designing an embedding that 
preserves distances at a single scale $r>0$,
while still maintaining a one-sided \lip condition for all scales.
We now state our most basic new result, which achieves only 
a constant distortion (for the desired scale).

\begin{theorem} \label{thm:simple1scale}
For every scale $r>0$ and every $0<\delta<1/4$,
every finite set $S\subset \ell_2$ admits an embedding
$\varphi:S\to \ell_2^k$ for $k = \tilde O(\log\tfrac1\delta\cdot(\dim S)^2)$,
satisfying:
\begin{enumerate} \compactify
\renewcommand{\theenumi}{(\alph{enumi})}
\item\label{it:lip}
\lip condition:
$\|\varphi(x)-\varphi(y)\|_2 \le \|x-y\|_2$
for all $x,y\in S$;
\item\label{it:bilip}
Bi-\lip at scale $r$:
$\|\varphi(x)-\varphi(y)\|_2 = \Omega(\|x-y\|_2)$
whenever $\|x-y\|_2\in [\delta r,r]$; and
\item\label{it:threshold}
Boundedness: $\|\varphi(x)\|_2 \leq r$ for all $x\in S$.
\end{enumerate}
\end{theorem}

The constant factor accuracy achieved by this theorem is too weak
to achieve the $1+\veps$ distortion asserted in Theorem \ref{thm:main}.
While we cannot improve condition \ref{it:bilip} to a factor of $1+\veps$,
we are able to refine it in a useful way. Roughly speaking, we introduce a 
``correction'' function $\tilG:\R\to\R$,
such that whenever $\|x-y\|_2\in[\delta r,r]$,
\begin{equation}  \label{eq:tilG}
  \frac{\|\varphi(x)-\varphi(y)\|_2} {\|x-y\|_2}
  = (1\pm\veps)\ \tilG(\tfrac{\|x-y\|_2}{r}).
\end{equation}
This function $\tilG$ does not depend on $r$,
and it equals $\Theta(1)$ in the range $[\delta,1]$ for fixed $\delta$. 
Using the correction 
function, we obtain very accurate bounds on distances in the target space,
at the price of increasing the dimension by a factor of $\tilde O(1/\veps^3)$.
This high-level idea is implemented in Theorem \ref{thm:OptDistortion},
which immediately implies Theorem \ref{thm:simple1scale},
although the precise bound therein slightly differs from 
Equation \eqref{eq:tilG}.

Embeddings for a single scale are commonly used in the embeddings literature, 
though not in the context of dimension reduction.
It is plausible that in some applications, 
a single-scale embedding may suffice, or even provide better bounds 
than our snowflake embedding (or Question \ref{q2}). 

\paragraph{Extension to {\boldmath $\ell_1$}.}

Our dimension reduction framework extends also to $\ell_1$ 
(i.e.\ $S\subset \ell_1$ and $\Phi:S\to \ell_1^k$) 
as discussed in 
Section \ref{sec:lp}.
%
The bounds we obtain therein are worse than in the $\ell_2$ case,
namely the dimension $k$ is doubly exponential in $\dim(S)$.
We remark that previous work on dimension reduction in $\ell_p$ spaces
\cite{JL84,Schechtman87, Talagrand90, Ball90, Talagrand95, Matousek96}
(with the exception of the $\ell_\infty$ snowflake of Har-Peled and Mendel \cite{HM06})
did not establish any dimension bound in terms of $\lambda(S)$;
these bounds are all expressed in terms of $n=|S|$, or of the 
dimension of $S$ as a linear subspace.

For ultrametrics, our framework provides even stronger bounds,
which resolve Question \ref{q2} in the affirmative, as follows.
Ultrametrics embed isometrically (i.e.\ with distortion 1) into $\ell_2$, 
hence Theorem \ref{thm:main} is immediately applicable.
We can then eliminate the snowflake operator (i.e.\ achieve $\alpha=1$)
by the observation that 
$(M,d)$ is an ultrametric if and only if $(M,d^2)$ is an ultrametric,
and thus Theorem \ref{thm:main} is applicable to the ultrametric $d^2$ 
with $\alpha=1/2$.
Moreover, the dimension bound can be improved by replacing some steps
with more specialized machinery.
However, in retrospect a near-optimal bound may be obtained 
by minor refinements of \cite[Lemma 12]{ABN09}.

\paragraph{Technical contribution.}

The main technical challenge is to keep both distortion and dimension 
under tight control. 
We use a relatively large number of the tools developed recently 
in the metric embeddings literature,  
combining them in a manner that yields 
a rather strong outcome ($1+\veps$ distortion).
Several of the tools we use are nonlinear, hence 
our approach can potentially be used to circumvent the limitation on 
linear embeddings pointed out by \cite{IN07}.

Our results may also be viewed as partial progress towards Question \ref{q2}: 
Observe that Theorem \ref{thm:main} answers that question positively 
in the special case where also the square of the given metric 
is known to be Euclidean (e.g. ultrametrics satisfy this condition).
Further, Theorem \ref{thm:simple1scale} achieves 
bounds that relax those required by Question \ref{q2}.
Moreover, if the answer to Question \ref{q2} is negative (which is not 
unlikely), then our results may be essentially the closest alternative.

\subsection{Related work} 
A summary of some related work on embeddings, meant to put our
results in context, is found in 
Table \ref{tab:related}. Subsequent to the public posting of this paper,
the authors of \cite{BRS11} concluded that an extension of their earlier work
from 2007, coupled with the 
framework presented here, yields a stronger version of Theorem \ref{thm:main}, where 
the target dimension is improved to 
$\tilde{O}(\veps^{-3}\tilde{\alpha}^{-2}\dim S)$. 
This additional result has been appended to the most recent version of \cite{BRS11}.
Later, Bartal and Gottlieb \cite{BG14} presented a new single-scale embedding for all
$\ell_p$, $1 \le p < 2$, and used our framework to derive a snowflake embedding for
$\ell_p$ with only polynomial dependence on the doubling dimension.

\begin{table}
\centering
\begin{tabular}{| l  | l  l | l  l l |}
\hline
Reference & Origin & Target space & Distortion & Dimension & Snowflake $\alpha$\\
\hline
\hline
\cite{JL84}	& $\ell_2$ 	& $\ell_2$ 	& $1+\veps$	& $O(\veps^{-2}\log n)$ 	& none ($\alpha=1$)\\
\cite{Assouad83} & doubling 	& $\ell_2$ 	& $2^{O(\dim S)}$ & $2^{O(\dim S)}$		& fixed $\alpha<1$	\\	
\cite{GKL03}	& doubling	& $\ell_2$	& $\tO(\dim S)$	&  $\tO(\dim S)$ & fixed $\alpha<1$	\\
\cite{HM06}	& doubling	& $\ell_\infty$	& $1+\veps$ & $\veps^{-O(\dim S)}$	& fixed $\alpha<1$ \\
\cite{ABN08}	& doubling	& $\ell_p$, $p\ge 1$	& $O(\log^{1+\veps} n)$ & $O(\veps^{-1}{\dim S})$ & none \\
Theorem \ref{thm:main}	& $\ell_2$	& $\ell_2$	& $1+\veps$  & $\tilde{O}(\veps^{-4} \dim^2 S)$ & fixed $\alpha<1$ \\
\hline
\end{tabular}
\caption{A sampling of related work; holds for arbitrary $\veps\in(0,1)$}
\label{tab:related}
\end{table}

\subsection{Applications}

In many settings, data is provided as points in $\ell_p$, 
and it is extremely advantageous to represent the data using 
a low-dimensional space.
For instance, the cost of many data processing tasks
(in terms of runtime, storage or accuracy) grows exponentially 
with the embedding dimension.
In many such cases, our machinery can reduce the embedding dimension 
to close to the data's doubling dimension,
leading to significant performance improvement.
This approach is suitable for problems
(i) that depend on pairwise distances but can tolerate small distortion; 
and 
(ii) whose algorithms depend heavily on the embedding dimension, so that
the improved performance given by the lower dimension outweighs the
overhead cost of computing the dimensionality reduction.

We provide in Section \ref{sec:app} two examples where our 
dimensionality reduction results have immediate algorithmic applications.
The first one is an approximate Distance Labeling Scheme, 
where the main complexity measure is the storage required at each network node.
The second example is approximation algorithms for clustering algorithms,
where running time is typically exponential in the dimension.
In both cases, the final approximation obtained is $1+\veps$.

On a more conceptual level, our embeddings may explain a common 
empirical phenomenon regarding low-dimensional data: Many
heuristics that represent (non-Euclidean) input data as points in 
Euclidean space find that low-dimensional 
Euclidean space is sufficient to yield a fair representation, see 
e.g. \cite{NZ02} for networking and \cite{TSL00,RS00} for machine learning. 
Our results can be interpreted as conveying the following principle: 
Intrinsically low-dimensional data that admits a meaningful 
representation in $\ell_2$, can actually be represented in 
\emph{low-dimensional} $\ell_2$.

\paragraph{Implementation.} 

All our embedding results are algorithmic --- 
they are constructive and can be computed in polynomial time. 
The details are mostly straightforward, and we do not address 
this issue explicitly. It is possible that the running time may
be improved further, and perhaps even be brought close to linear. 
(For example, the Gaussian transform is computed quickly via the 
Gram matrix.) Two nontrivial steps in this direction are the 
implementation of Kirszbraun's Theorem, which is usually solved
as a semidefinite program, and our use of padded partitions, 
which require an application of the \Lovasz Local Lemma.

\section{Preliminaries and tools}
\label{sec:prelim}


\noindent{\bf Doubling dimension.}
For a metric $(X,d)$, let $\lambda$ be the infimum value such that the points of
every ball in $X$ can be covered by $\lambda$ balls of half the radius.
(Here, a ball may be centered at any point in the ambient space in which $X$ resides.) 
The {\em doubling dimension} of $X$ is dim$(X)=\log_2\lambda$. A metric is {\em 
doubling} when its doubling dimension is finite. The following property 
can be demonstrated via a repetitive application of the doubling property \cite{KL04}.

\begin{property}\label{prop:doubling}
For set $X$ with doubling constant $\lambda$, if the minimum 
interpoint distance in any $S \subset X$ is at least $\alpha$, 
and the diameter of $S$ is at most 
$\beta$, then $|S| \le \lambda^{O(\log (\beta/\alpha))}$.
\end{property}


\noindent{\bf $\epsilon$-nets.}
For a point set $S$, an $\epsilon$-net of $S$ is a subset $T \subset S$ with the following
properties:
(i) Packing: For every pair $u,v \in T$, $d(u,v) \ge \epsilon$.
(ii) Covering: Every point $u \in S$ is strictly within distance 
$\epsilon$ of some point $v \in T$: $d(u,v) < \epsilon$.

\paragraph{Lipschitz norm.}
Let $(X,d_X)$ and $(Y,d_Y)$ be metric spaces.
A function $f:X\to Y$ is said to be {\em $K$-\lip} (for $K>0$)
if for all $x,x'\in X$ we have $d_Y(f(x),f(x')) \le K\cdot d_X(x,x')$.
The {\em \lip constant} (or \lip norm) of $f$, denoted $\lipnorm{f}$,
is the infimum over $K>0$ satisfying the above.
A $1$-\lip function is called in short \lip.
We recall the following basic property of Lipschitz functions:
Let $f:X\to \ell_2^k$ and $g:X\to\R$.
Then their product $fg:x\to g(x)\cdot f(x)$ has \lip norm
\[
  \lipnorm{fg} \leq
  \lipnorm{f}\cdot \max_x |g(x)| + \lipnorm{g}\cdot \max_x \|f(x)\|.
\]

\remove{
We enumerate in Appendix \ref{sec:lip-prop} some immediate properties 
of Lipschitz functions that we will need.
For example, 
for all $f_1,\ldots,f_m:X\to \ell_2^k$, the sum
$\sum_{i=1}^m f_i$ which maps $x\to f_1(x)+\ldots+f_m(x)\in\ell_2^k$
has \lip norm
$\lipnorm{\sum_{i=1}^m f_i}
  \leq (\sum_{i=1}^m \lipnorm{f_i})^{1/2}$.
}

\smallskip
\noindent {\bf Extension Theorem.}
The Kirszbraun Theorem \cite{Kir34} states that if $S$ and $X$ are Euclidean spaces, $T \subset S$, and
there exists a Lipschitz function $f: T \to X$; then there exists a function
$\tilde{f}: S \to X$ that has the same Lipschitz constant as $f$ and also {\em extends} $f$, i.e.\ $\tilde{f}|_T = f$,
meaning that the restriction of $\tilde{f}$ to $T$ is identical to $f$.


\smallskip 
\noindent {\bf Bounded distances and the Gaussian Transform.} 
A \emph{metric transform} maps a distance function to another 
distance function on the same set of points (e.g. maps $(X,d)$ to 
$(X,d^{1/2})$).
We say that a metric transform is \emph{bounded} (by $T>0$) if it always 
results with a distance function where all interpoint distances are 
bounded (by $T>0$).
The \emph{Gaussian transform} is a metric transform that maps value $t$ to 
$G_r(t)=r(1-e^{-t^2/r^2})^{1/2}$, where $r>0$ is a parameter. 
Schoenberg \cite{Sch38a,Sch38b,DL97} showed that 
the Gaussian transform maps Euclidean spaces to Euclidean spaces. 
That is, for every $r>0$ and $X\subset L_2$ there is an embedding 
$g:X \to L_2$ such that for all $x,y\in X$ 
we have $\|g(x)-g(y)\|_2= G_r(\|x-y\|_2)$.
It is easily verified that 
\begin{equation} \label{eq:Glip}
  G_r(t) \le t,\qquad \forall t\ge 0,
\end{equation}
and thus $\lipnorm{g} \le 1$.
In addition, $G_r(t)\le r$ for all $t$, 
hence the Gaussian transform is bounded.
The idea of using bounded transforms for embeddings is due to 
\cite{BRS11}.

\paragraph{Probabilistic partitions.} 

Probabilistic partitions are a common tool used in embeddings. Let $(X,d)$ 
be a finite metric space. A partition $P$ of $X$ is a collection of 
non-empty pairwise disjoint clusters $P=\{C_1,C_2,\ldots,C_t\}$ such 
that $X=\cup_jC_j$. For $x\in X$ we denote by $P(x)$ the cluster 
containing $x$.

We will need the following decomposition lemma due to Gupta, Krauthgamer 
and Lee \cite{GKL03}, Abraham, Bartal and Neiman \cite{ABN08}, and Chan, 
Gupta and Talwar \cite{CGT10}. Let 
$B(x,r) = \{y| \, \|x-y\| \le r \}$.

\begin{theorem}[Padded Decomposition of doubling metrics \cite{GKL03,ABN08,CGT10}] 
\label{thm:decomp}
There exists a constant $c_0>1$,
such that for every metric space $(X,d)$, every $\veps\in(0,1)$, 
and every $\Delta > 0$, there is a
multi-set $\mathcal D = [P_1, \ldots, P_m]$ of partitions of $X$,
with $m \leq c_0 \veps^{-1} \dim(X)\log \dim(X)$, such that
\begin{enumerate} \compactify
\item Bounded radius: $\diam(C) \leq \Delta$ for all clusters $C \in \bigcup_{i=1}^m P_i$.
\item Padding: If $P$ is chosen uniformly {from} $\mathcal D$, then for all $x \in X$,
$$
\Pr_{P\in\mathcal D} [B(x,\tfrac{\Delta}{c_0\dim(X)}) \subseteq P(x)] \geq 1-\veps.
$$
\end{enumerate}
\end{theorem}
Remark: \cite{GKL03} provided slightly different quantitative bounds
than in Theorem \ref{thm:decomp}. 
The two enumerated properties follow, for example, from Lemma 2.7 in 
\cite{ABN08}, and the bound on support-size $m$ follows by an application 
of the \Lovasz\ Local Lemma sketched therein.

\section{Dimension Reduction for {\boldmath $\ell_2$}}
\label{sec:l2} 

In this section we first design a single scale embedding
that achieves distortion $1+\veps$ after including a correction function.
This result is stated in Theorem \ref{thm:OptDistortion} below,
which is a refined version of Theorem \ref{thm:simple1scale}.
We then use this single scale embedding to prove Theorem \ref{thm:main}
in Section \ref{sec:l2betterSnowflake}.
Throughout this section, the norm notation $\|\cdot\|$ denotes $\ell_2$-norms.
We make no attempt to optimize constants.
Following Section \ref{sec:prelim},
define $G:\R\to\R$ by $G(t)=(1-e^{-t^2})^{1/2}$,
and let
$$
   G_r(t)=r\cdot G(t/r)=r(1-e^{-t^2/r^2})^{1/2}.
$$

\begin{theorem} \label{thm:OptDistortion}
For every scale $r>0$ and every $0<\delta,\veps<1/4$,
every finite set $S\subset \ell_2$ admits an embedding 
$\varphi:S\to \ell_2^k$ for $k=\tilde O(\veps^{-3}\log\tfrac1\delta\cdot (\dim S)^2)$,
satisfying:
\begin{enumerate} \compactify
\renewcommand{\theenumi}{(\alph{enumi})}
\item\label{it:lipOpt}
\lip condition:
$\|\varphi(x)-\varphi(y)\| \le \|x-y\|$
for all $x,y\in S$.
\item\label{it:bilipOpt2}
$1+\veps$ distortion to the Gaussian (at scales near $r$):
  For all $x,y\in S$ with $\delta r \le \|x-y\|\le \frac{r}{\delta}$,
$$
  \frac{1}{1+\veps}
  \le \frac{\|\varphi(x)-\varphi(y)\|}{G_r(\|x-y\|)} \le 1;
$$
\item\label{it:thresholdOpt}
Boundedness: $\|\varphi(x)\| \leq r$ for all $x\in S$.
\end{enumerate}
\end{theorem}

We note that to prove Theorem \ref{thm:OptDistortion}, 
it in fact suffices to demonstrate an embedding $\varphi'$
with the same dimension and boundedness guarantess, but 
which is $(1+a\veps)$-\lip and guarantees $(1+b\veps)$ distortion to the Gaussian
(for arbitrary constants $a,b$).
We scale down $\varphi'$ by a factor of
$(1+a\veps)$, so that resulting function is 1-\lip. 
Then $\varphi$ is recovered by taking the scaled $\varphi'$ with parameter 
$\veps' = (1+a\veps)(1+b\veps) - 1 = \Theta(\veps)$:
This yields $(1+\veps)$ distortion to the Gaussian, 
and has no asymptotic effect on the dimension.


\subsection{Embedding for a single scale}
\label{sec:l2BetterSingleScale}

Our construction of the embedding $\varphi$ for Theorem \ref{thm:OptDistortion}
proceeds in seven steps, as described below.
Let $\lambda=\lambda(S)$.
All the hidden constants are absolute, 
i.e.\ independent of $|S|$, $\lambda$, $\veps$, $\delta$ and $r$. 
It is plausible that the dependence of target dimension on $\log \lambda$
can be improved to be near-linear, by carefully combining some of these steps.

\begin{description} \compactify
\item[Step 1 (Net Extraction):]
Let $N\subseteq S$ be an $(\veps \delta r)$-net in $S$.

\item[Step 2 (Padded Decomposition):]
Compute for $N$ a padded decomposition with padding $\frac{3r}{\delta}$.
More specifically,
by Theorem~\ref{thm:decomp}, there is a multiset $[P_1,\ldots,P_m]$
of partitions of $N$, where every point is 
$\frac{3r}{\delta}$-padded in $1-\veps$ fraction of the partitions,
all clusters have diameter bounded by 
$\Delta
= c_0 \frac{r}{\delta} \log\lambda
=O(\frac{r}{\delta}\log\lambda)$,
and 
$m 
= c_0 \veps^{-1}\log \lambda \log \log \lambda
= O(\veps^{-1}\log \lambda \log \log \lambda)$.

\item[Step 3 (Bounding Distances):]
In each partition $P_i$ and each cluster $C\in P_i$,
bound the interpoint distances in $C$ at maximum value $r$.
Specifically, using a Gaussian transform as per Section 
\ref{sec:prelim}, 
obtain a map $g_C: C \to \ell_2$ such that
\[
  \|g_C(x)-g_C(y)\|_2^2 
  = G_r(\|x-y\|)^2 
  = r^2(1-e^{-\|x-y\|_2^2/r^2}),\quad \forall x,y\in C.
\]

\item[Step 4 (Dimension Reduction):]
For each partition $P_i$ and each cluster $C\in P_i$,
the point set $g_C(C)\in\ell_2$ admits a dimension reduction,
with distortion $1+\veps$.
Specifically, by the JL-Lemma there is a map
$\Psi_{\JL}: g_C(C)\to \ell_2^{k'}$
such that
\begin{equation} \label{eq:singleScaleJL}
  \tfrac{\|t-t'\|}{1+\veps} \le \| \Psi_{\JL}(t)-\Psi_{\JL}(t') \| \le \|t-t'\|,
  \quad \forall t,t'\in g_C(C),
\end{equation}
and the target dimension is (using Property \ref{prop:doubling})
\begin{align*}
k' &=O(\veps^{-2} \log |C|)	\\
   &= O(\veps^{-2} \log (\lambda^{O(\log (\Delta/\veps \delta r))}))	\\
   &= O(\veps^{-2} [\log \left(\frac{\log \lambda}{\veps \delta^2} \right) \log \lambda])	\\
   &= O(\veps^{-2} \log\tfrac{1}{\veps\delta} \cdot \log\lambda \log\log\lambda).
\end{align*}
Composing the last two steps, define
$f_C=\Psi_{\JL} \circ g_C$ mapping $C\to \ell_2^{k'}$.

\item[Step 5 (Gluing Clusters):]
For each partition $P_i$, ``glue'' the cluster embeddings $f_C$
by smoothing them near the boundary.
Specifically, for each cluster $C\in P_i$,
assume by translation that $f_C$ attains the origin,
i.e.\ there exists $z_C\in C$ such that $\|f_C(z_C)\|=0$.
Define $h_C:C\to \R$ by $h_C(x) = \min_{y\in N\setminus C} \|x-y\|$,
as a proxy for $x$'s distance to the boundary of its cluster.
Now define $\varphi_i:N\to \ell_2^{k'}$ by
\[
  \varphi_i(x) = f_{P_i(x)}(x) \cdot \minn{1,\ \tfrac{\delta}{r} h_{P_i(x)}(x)};
\]
recall that $P_i(x)$ is the unique cluster $C\in P_i$ containing $x$.

\item[Step 6 (Gluing Partitions):]
Combine the maps obtained in the previous step via direct sum and scaling.
Specifically, define $\varphi:N\to\ell_2^{mk'}$ by
$\varphi = m^{-1/2} \bigoplus^m_{i=1} \varphi_i$.

\item[Step 7 (Extension beyond the Net):]
Use the Kirszbraun theorem to extend the map $\varphi$ to all of $S$,
without increasing the \lip constant.
\end{description}

\subsection{Proof of Theorem \ref{thm:OptDistortion}}
\label{sec:l2BetterSingleScaleAnalysis}

Let us show the embedding $\varphi$ constructed above indeed satisfies 
the conclusion of Theorem \ref{thm:OptDistortion}. 
By construction, the target dimension is
$mk'
  = O(\veps^{-3} \log\tfrac{1}{\veps\delta} (\log\lambda \log\log\lambda)^2)
$.

We first focus on points in the net $N$, 
and later extend the analysis to all points in $S$.
Let us start with a few observations, 
Lemmata \ref{lem:boundG} and \ref{lem:singleScale}.

\begin{lemma} \label{lem:boundG}
For $r,t > 0$,
\begin{enumerate} \compactify
\renewcommand{\theenumi}{(\roman{enumi})}
\item
$G_r(t)$ is monotone increasing in $t$ and $r$.
\item
$\frac{G_r(t)}{t}$ is monotone decreasing in $t$
\item
Let $0<\eta<1/3$ and suppose $0< t'\le (1+\eta)t$.
Then $\frac{G_r(t')}{G_r(t)} \le 1+3\eta$.
\end{enumerate}
\end{lemma}

The proof of Lemma \ref{lem:boundG} is found in Appendix \ref{sec:A}.

\begin{lemma} \label{lem:singleScale}
For every $x,y\in N$ and every $i\in\{1,\ldots,m\}$,
\begin{enumerate} \compactify
\renewcommand{\theenumi}{(\roman{enumi})}
\item \label{it:part1}
$\|f_{P_i(x)}(x)\| \le r$.
\item \label{it:part2}
If $P_i(x)=P_i(y)=C$ then 
$\|f_C(x)-f_C(y)\| \le G_r(\|x-y\|) \le \|x-y\|$.
\item \label{it:part3}
If $P_i(x)\neq P_i(y)$ then $h_{P_i(x)}(x) \leq \|x-y\|$.
\end{enumerate}
\end{lemma}
\begin{proof}
For assertion \ref{it:part1}, recall that by the translation of Step 5, 
every cluster, and in particular $C=P_i(x)$, contains a point $z_C\in N$
such that $f_C(z_C)=0$. 
Thus, using Equation \eqref{eq:singleScaleJL} and noting that 
$G_r(t) \le r$, we have
\[ \|f_C(x)-f_C(z_C)\| 
  \le \|g_C(x)-g_C(z_C)\| 
  =  G_r(\|x-z_C\|)
  \le r.
\]
To prove the assertion \ref{it:part2}, 
use Equations \eqref{eq:Glip} and \eqref{eq:singleScaleJL}, to get
\[ 
  \|f_C(x)-f_C(y)\|
  \le \|g_C(x)-g_C(y)\|
  \le G_r(\|x-y\|)
  \le \|x-y\|.
\]
For assertion \ref{it:part3}, since $C = P_i(x)\neq P_i(y)$ we have that
$y \in N \setminus C$, and so 
$h_C(x) = \min_{z\in N\setminus C} \|x-z\| \le \|x-y\|$.
\end{proof}

\paragraph{Analysis for the net $N$.}
We now prove assertions (a)-(c) of Theorem \ref{thm:OptDistortion}
for (only) net points. 
(We shall need this later to complete the proof of the theorem.)
To this end, fix $x,y\in N$.

\begin{enumerate}
\item[\ref{it:lipOpt}] \lip:
If $\|x-y\| > \frac{r}{\delta}$, we use the boundedness condition 
(assertion (c) of Theorem \ref{thm:OptDistortion}) and the fact that 
$\delta \le \frac{1}{4}$ to get
\[
  \|\varphi(x) - \varphi(y)\| 
  \le \|\varphi(x)\| + \|\varphi(y)\| 
  \le 2r 
  < \tfrac{r}{\delta} \le \|x-y\|.
\]
Assume now that $\|x-y\| \le \frac{r}{\delta}$. Then by Step 6
\begin{equation} \label{eq:singleScaleCombine}
  \|\varphi(x)-\varphi(y)\|^2
  = \tfrac{1}{m} \sum_{i=1}^m   \|\varphi_i(x)-\varphi_i(y)\|^2.
\end{equation}
To bound the righthand side, fix $i\in\{1,\ldots,m\}$
and consider separately the following three cases.
\begin{description}
\item[Case 1:] $x$ is padded.
The padding is $\tfrac{3r}{\delta}$, 
hence $x$ and $y$ belong to the same cluster $C=P_i(x)=P_i(y)$.
This padding implies that $h_{C}(x) \ge \tfrac{3r}{\delta}$
(by definition, see Step 5), and thus 
$\varphi_i(x) = f_{C}(x) \cdot \min \{1,h_{C}(x) \} = f_{C}(x)$.
Similarly, the triangle inequality implies
$h_{C}(y) \ge h_C(x) - \|x-y\| \ge \tfrac{2r}{\delta}$, and
thus $\varphi_i(y) = f_{C}(y)$ as well.
Using Lemma \ref{lem:singleScale}\ref{it:part2}
\[
  \|\varphi_i(x)-\varphi_i(y)\|
  = \|f_C(x)-f_C(y)\| 
  \le \|x-y\|.
\]

\item[Case 2:] $x$ is not padded and $P_i(x)\neq P_i(y)$.
By Lemma \ref{lem:singleScale}\ref{it:part1}\ref{it:part3},
\[
  \|\varphi_i(x)\|
  \le  \| f_{P_i(x)}(x)\| \cdot \tfrac{\delta}{r} h_{P_i(x)}(x)
  \le \delta h_{P_i(x)}(x)
  \le \delta \|x-y\|.
\]
A similar bound holds for $\varphi_i(y)$, and we obtain
\[
  \|\varphi_i(x) - \varphi_i(y)\|
  \leq \|\varphi_i(x)\| + \|\varphi_i(y)\|
  \leq 2\delta\|x-y\|
  \leq \|x-y\|.
\]

\item[Case 3:] $x$ is not padded and $x,y$ belong to the same cluster 
$P_i(x)=P_i(y)=C$.
Restrict $\varphi_i$ to $C$ and 
write it as the product of the two functions $z\mapsto f_{C}(z)$ 
and $\tilh_C:z \mapsto \minn{1,\ \tfrac{\delta}{r} h_C(z)}$.
It follows that
\[
  \frac{\|\varphi_i(x) - \varphi_i(y)\|}{\|x-y\|}
  \le \lipnorm{f_C}\cdot \max_{x\in C} |\tilh_C(x)|
  + \lipnorm{\tilh_C} \cdot \max_{x\in C} \|f_C(x)\|
\]
Now $\lipnorm{f_C} \le 1$ and also $\lipnorm{h_C} \le 1$
(where the first assertion follows the fact that
$\lipnorm{g} \le 1$ in conjunction with Equation \eqref{eq:singleScaleJL},
and the second as a consequence of the triangle inequality).
By definition $\max_{z\in C} |\tilh_C(z)| \le 1$, and 
it is easy to verify that 
$\lipnorm{\tilh_C} 
  \le \tfrac{\delta}{r}\cdot \lipnorm{h_C} 
  \le \tfrac{\delta}{r}$.
Plugging in these estimates and the bound on $f_C$ obtained from 
Lemma \ref{lem:singleScale}(i), we have 
\[
  \frac{\|\varphi_i(x) - \varphi_i(y)\|}{\|x-y\|}
  \le 1\cdot 1 + \tfrac{\delta}{r}\cdot r 
  = 1+\delta.
\]
\end{description}
Now combine these three cases by plugging into 
Equation \eqref{eq:singleScaleCombine}.
Since $x$ is padded in at least $1-\veps$ fraction of the partitions $P_i$,
and for the remaining partitions we can use the worst bound among the 
three cases, we get
\[
  \|\varphi(x)-\varphi(y)\|^2
  \leq (1-\veps)\|x-y\|^2 + \veps (1+\delta)^2 \|x-y\|^2
  < (1+3\veps\delta) \|x-y\|^2.
\]

\item[\ref{it:bilipOpt2}] Distortion to the Gaussian:
We prove the result for a slightly extended range
$\tfrac12\delta r\le \|x-y\| \le \tfrac{2r}{\delta}$.
Recalling that $\delta < \frac{1}{4}$ 
and $\frac{G_r(t)}{t}$ is monotonically decreasing in $t$
(Lemma \ref{lem:boundG}(ii)),

\begin{equation} \label{eq:boundGr}
  \frac{G_r(\|x-y\|)}{\|x-y\|} 
  \ge \frac{G_r(2r/\delta)}{(2r/\delta)}
  >   \frac{G_r(8r)}{(2r/\delta)}
  > \frac{\delta}{3},
\end{equation}

We proceed by considering the exact same three cases as above.
\begin{description}
\item[Case $1'$:] $x$ is padded. 
By the analogous case above, $P_i(x)=P_i(y)=C$ and 
\[
  \|\varphi_i(x)-\varphi_i(y)\| =
  \|f_C(x)-f_C(y)\|.
\]
By \eqref{eq:singleScaleJL} we have
$
    1-\veps < \frac{1}{1+\veps} \le \frac{\|f_C(x)-f_C(y)\|}{\|g_C(x)-g_C(y)\|} \leq 1
$,
where, by construction, the denominator equals $G_r(\|x-y\|)$.
Altogether, we get
\[
  1-\veps \le \frac{\|\varphi_i(x)-\varphi_i(y)\|}{G_r(\|x-y\|)} \le 1.
\]

\item[Case $2'$:] $x$ is not padded and $P_i(x)\neq P_i(y)$.
Combining the analogous case above and Equation \eqref{eq:boundGr}, we 
have
\[
  \|\varphi_i(x) - \varphi_i(y)\|
  \leq 2\delta\|x-y\|
  < 6G_r(\|x-y\|).
\]

\item[Case $3'$:] $x$ is not padded and $x,y$ belong to the same cluster 
$P_i(x)=P_i(y)=C$.
Refining the analysis in the analogous case above
and using Equation \eqref{eq:boundGr} and the triangle inequality, 
we have
\begin{align*}
  \|\varphi_i(x) - \varphi_i(y)\|
  &= \|f_C(x)\tilh_C(x) - f_C(y)\tilh_C(y)\| \\
  &\le \|f_C(x)\tilh_C(x)-f_C(x)\tilh_C(y)\| + \|f_C(x)\tilh_C(y) - f_C(y)\tilh_C(y)\| \\
  &\le \|f_C(x)\| \cdot |\tilh_C(x)-\tilh_C(y)| + \|f_C(x) - f_C(y)\|\cdot |\tilh_C(y)| \\
  &\le r\cdot \tfrac{\delta}{r}\|x-y\|+ G_r(\|x-y\|) \cdot 1\\
  &\le 4G_r(\|x-y\|).
\end{align*}
\end{description}

Again combine these three cases by plugging into Equation 
\eqref{eq:singleScaleCombine}
and recalling that $x$ is padded in at least $1-\veps$ fraction of partitions;
we thus get 
\begin{equation} \label{eq:boundGr2}
  (1-\veps)^3
  \le \frac{\|\varphi(x)-\varphi(y)\|^2}{G_r(\|x-y\|)^2}
  \le (1-\veps) + \veps\cdot 36
  =   1+35\veps.
\end{equation}

For later use, recall that $\veps < \frac{1}{4}$, and let us record that 
\begin{equation} \label{eq:boundGr3}
\|\varphi(x)-\varphi(y)\| 
\le G_r(\|x-y\|)(1+35\epsilon)^{1/2} 
< G_r(\|x-y\|)(1+18\epsilon) 
< \tfrac{11}{2}G_r(\|x-y\|).
\end{equation}

\item[\ref{it:thresholdOpt}.] Boundedness:
By the fact $0\le h_{P_i(x)}(x)\leq 1$ and 
Lemma \ref{lem:singleScale}\ref{it:part1},
\[
  \|\varphi(x)\|^2
  \le \tfrac{1}{m} \sum_{i=1}^m \|\varphi_i(x)\|^2
  \le \tfrac{1}{m} \sum_{i=1}^m \| f_{P_i(x)}(x)\|^2
  \le r^2.
\]
\end{enumerate}
This completes the analysis for net points $x,y\in N$. 

\remove{
\paragraph{Analysis for entire $S$.} 
To complete the proof of Theorem \ref{thm:OptDistortion}, 
we need to prove assertions (a)-(c) for all points in $S$. 
This follows via standard arguments using the above proof for net points 
and appealing to the Kirszbraun theorem to argue that $\varphi$ is 
$1$-Lipschitz on the entire $S$. For technical reasons, it also  
requires some careful estimates on $G_r(t)$. 
This part of the proof is relegated to Appendix \ref{sec:entire-S}.

\paragraph{Analysis for entire $S$.}
We extend the above analysis to all points in $S$ 
in Appendix \ref{app:entireS}. 
Roughly speaking, for every $x,y\in S$, we let $x',y' \in N$ be the 
net points closest to $x$ and $y$, respectively,
and use the bounds obtained above for $x',y'$.
}

\paragraph{Analysis for entire $S$.} 
We extend the previous analysis for net points 
to all points in $S$. 
Fix $x,y\in S$, and let $x',y' \in N$ be the net points closest to $x$ 
and $y$, respectively. 
Recalling that $N$ is an $\veps \delta r$-net, 
we have $\|x-x'\|,\|y-y'\| \le \veps \delta r$. 
To prove the \lip requirement, 
recall that Step 8 extends $\varphi$ from the net $N$ to the entire $S$ 
using the Kirszbraun theorem, i.e.\ without increasing its \lip\ norm, hence
\[
  \|\varphi(x)-\varphi(y)\| \le \|x-y\|.
\]
Using this \lip condition and the triangle inequality, 
we immediately obtain the boundedness requirement: 
\[
  \|\varphi(x)\| 
  \le \|\varphi(x')\| + \lipnorm{\varphi}\|x-x'\|
  \le (1+\veps\delta)r.
\]

To prove the requirement of distortion to the Gaussian 
(which is slightly more involved)
consider the case where
$\delta r\le \|x-y\| \le \frac{r}{\delta}$.
By the triangle inequality,
\begin{equation}  \label{eq:1}
  \Big| \|x-y\| - \|x'-y'\| \Big| 
  \le \|x-x'\| + \|y-y'\|
  \le 2\veps\delta r.
\end{equation}
We conclude that
$ (1-2\veps)\delta r  
  \le \|x'-y'\|
  \le (\frac1\delta + 2\veps\delta)r
$.
Hence
$ \frac{\delta r}{2}
  < \|x'-y'\|
  < \frac{2r}{\delta} 
$
and the net-points $x',y'\in N$ must possess the bound for distortion to the Gaussian. 
It also follows that 
$2\veps \delta r \le 4(1-2\veps)\veps\delta r \le 4\veps \|x'-y'\|$.
Using the \lip condition on $\varphi$, and the above distortion to the 
Gaussian for net points (Equation (\ref{eq:boundGr3})), we similarly 
derive
\begin{equation}\label{eq:2}
  \Big| \|\varphi(x)-\varphi(y)\| - \|\varphi(x')-\varphi(y')\| \Big|
  \le \|x-x'\| + \|y-y'\|
  \le 2\veps\delta r
  \le 4\veps \|x'-y'\|
  <   22\veps G_r(\|x'-y'\|).
\end{equation}

Similar to the derivation of Equation \eqref{eq:1}, we derive $\|x'-y'\| 
\le (1+2\veps)\|x-y\|$, 
and by Lemma \ref{lem:boundG}(iii) we get 
$G_r(\|x'-y'\|) \le (1+6\veps)G_r(\|x-y\|)$.
Together with Equation \eqref{eq:2} and the upper bound for net points
(Equation (\ref{eq:boundGr3})), we obtain
\begin{align*}
  \|\varphi(x)-\varphi(y)\|
  &\le \|\varphi(x')-\varphi(y')\| + 22\veps G_r(\|x'-y'\|)\\
  &\le (1+18\veps+22\veps) G_r(\|x'-y'\|)\\
  &\le (1+40\veps)(1+6\veps) G_r(\|x-y\|).
\end{align*}
The other direction is analogous. 
By \eqref{eq:1} we have $\|x-y\| \le (1+4\veps)\|x'-y'\|$,
and by Lemma \ref{lem:boundG}(iii) we get 
$G_r(\|x-y\|) \le (1+12\veps)G_r(\|x'-y'\|)$.
Together with \eqref{eq:2} and the lower bound for net points (Equation 
(\ref{eq:boundGr2})), we obtain
\begin{align*}
  \|\varphi(x)-\varphi(y)\|
  &\ge \|\varphi(x')-\varphi(y')\| - 22\veps G_r(\|x'-y'\|)\\
  &\ge (1-\veps-22\veps)G_r(\|x'-y'\|) \\
  &\ge (1-23\veps)(1-12\veps) G_r(\|x-y\|).
\end{align*}

This completes the proof of Theorem \ref{thm:OptDistortion}.

\subsection{Snowflake Embedding}
\label{sec:l2betterSnowflake}

We now use Theorem \ref{thm:OptDistortion} (the single scale embedding)
to prove Theorem \ref{thm:main} (embedding for $d^\alpha$). 


Fix a finite set $S\subset\ell_2$ and $0<\veps<1/4$.
Assume without loss of generality that the minimum interpoint distance 
in $S$ is $1$.
Define 
$\tilde{\alpha} = \min \{ \alpha,1-\alpha\}$,
$p
=\frac{\ceil{\log_{1+\veps}(\frac{1}{\veps})}}{\tilde{\alpha}}
=O(\frac{1}{\tilde \alpha \veps} \log \frac{1}{\veps})
$ 
and the set 
$I=\{i\in\mathbb Z:\ (1+\veps)^{-2p} \le (1+\veps)^i \le (1+\veps)^{2p}\diam(S) \}$.
For each $i\in I$, let $\varphi_i:S\to\ell_2^k$ be the embedding that
achieves the bounds of
Theorem \ref{thm:simple1scale} for $S$ and $\veps$ with respect to parameters 
$r=(1+\veps)^i$ 
and
$\delta=(1+\veps)^{-p-1}=\Theta(\veps^{1/ \tilde \alpha})$.
Notice that each $\varphi_i$ has target dimension 
$k=\tilde O(\veps^{-3} \log\tfrac1\delta\cdot (\dim S)^2)
=\tilde O(\veps^{-3} \tilde{\alpha}^{-1} (\dim S)^2)$.

We shall now use the following technique due to Assouad \cite{Assouad83}.
First, each $\varphi_i$ is scaled by 
$r^{\alpha-1} = (1+\veps)^{i(\alpha-1)}$.
They are then grouped in a round robin fashion into $2p$ groups, 
and the embeddings in each group are summed up.
This yields $2p$ embeddings, each into $\ell_2^k$; these
are combined using a direct-sum, resulting in one map $\Phi$ into $\ell_2^{2pk}$.

Formally, let $i\equiv_{p} j$ denote that two integers $i,j$ are equal 
modulo $p$.
Define $\Phi:S\to \ell_2^{2pk}$ using 
the direct sum $\Phi=\bigoplus_{j\in[2p]} \Phi_j$, 
where each $\Phi_j:S\to\ell_2^k$ is given by 
$$
  \Phi_j=\sum_{i\in I:\ i\equiv_{2p} j} \frac{\varphi_i}{(1+\veps)^{i(1-\alpha)}}.
$$
For $M=M(\veps)>0$ that will be defined later, 
our final embedding is $\Phi/\sqrt M:S\to \ell_2^{2pk}$, which has target dimension 
$pk = \tilde O(\veps^{-4} \tilde{\alpha}^{-1} (\dim S)^2)$, as required.
It thus remains to prove the distortion bound. 
We will need the following lemma.

\begin{lemma} \label{lem:buckets}
Let $\Phi:S\to\ell_2^{2pk}$ be as above, let $x,y\in S$,
and define $B_i = \frac{\|\varphi_i(x)-\varphi_i(y)\|}{(1+\veps)^{i(1-\alpha)}}$.
Then for every interval $A\subset I$ of size $2p$
(namely $A=\{a-p,\ldots,a,\ldots,a+p-1\}$),
\begin{align*}
  \|\Phi(x)-\Phi(y)\|^2 
  &\le \sum_{i\in A} 
	\Big(B_i
        \  + \sum_{i'\in I\setminus A:\ i'\equiv_{2p} i} 
	B_{i'}
        \Big)^2 \\
  &= \sum_{i \in A} 
	\Big(B_i^2 +
        \ 2B_i  
	\sum_{i'\in I\setminus A:\ i'\equiv_{2p} i} B_{i'} +
	\Big(
	\sum_{i'\in I\setminus A:\ i'\equiv_{2p} i} B_{i'}
        \Big)^2 \Big) \\  
  &\le \sum_{a-p \le i < a} 
	\Big(B_i^2 + B_{i+p}^2 +
        \ 2(B_i + B_{i+p}) 
	\sum_{i'\in I\setminus A:\ i'\equiv_p i} B_{i'} +
	\Big(
	\sum_{i'\in I\setminus A:\ i'\equiv_p i} B_{i'}
        \Big)^2 \Big) \\  
  \|\Phi(x)-\Phi(y)\|^2 
  &\ge \sum_{i\in A} 
        \Big( \max \Big\{0,
	B_i
        \  - \sum_{i'\in I\setminus A:\ i'\equiv_{2p} i} 
	B_{i'}
        \Big\} \Big)^2 \\
  &\ge \sum_{a-p \le i < a} 
        \Big( B_i^2 + B_{i+p}^2 - 
	\ 2(B_i + B_{i+p})
	\sum_{i'\in I\setminus A:\ i'\equiv_p i} 
	B_{i'} 
	\Big) 
\end{align*}
\end{lemma}
\begin{proof}
By construction,
\[
  \Big\|\Phi(x)-\Phi(y)\Big\|^2 
  = \sum_{j\in[p]} \Big\|\Phi_j(x)-\Phi_j(y)\Big\|^2
  = \sum_{i\in A} \Big\| \sum_{i'\in I:\ i'\equiv_{2p} i}
            \frac{\varphi_{i'}(x)-\varphi_{i'}(y)}{(1+\veps)^{i(1-\alpha)}}
  \Big\|^2.
\]
Fix $i\in A$ and let us bound the term corresponding to $i$.
The first required inequality now follows by separating
(among all $i'\in I$ with $i'\equiv_{2p} i$)
the term for $i'=i$ from the rest,
and applying the triangle inequality for 
vectors $v_1,\ldots,v_s\in \ell_2^k$, namely, 
$
  \|\sum_l v_l\| \leq \sum_l \|v_l\|
$.
We then bound the sum of the terms for indices $i$ and $i+p$.
The second inequality follows similarly by separating the term for $i'=i$ 
from the rest, and applying the following triangle inequality for 
vectors $u,v_1,\ldots,v_s\in \ell_2^k$, namely, 
$
  \|u+\sum_l v_l\| \ge \max\{0, \|u\| - \sum_l \|v_l\| \}
$.  
We then bound the sum of the terms for indices $i$ and $i+p$. (Note
that $(\max\{0,b-c\})^2 \ge b^2 - 2bc$.)
\end{proof}

The proof of Theorem \ref{thm:main} proceeds by demonstrating that, for an 
appropriate choice of $A$ (meaning $a$, having fixed $p$), the leading terms in 
the above summations ($B_i^2$ and $B_{i+p}^2$ for $a \le i < a+p$) dominate 
the sum of all other terms of the summations. Fix $x,y\in S$, and let $i^*\in 
I$ be such that $(1+\veps)^{i^*} \le \|x-y\| \le (1+\veps)^{i^*+1}$. We 
wish to apply Lemma \ref{lem:buckets}. To this end, let $a=i^*$ and so
$A=\{i^*-p,\ldots,i^*+p-1\}$. Consider $i\in A$; we have the following 
lemma:

\begin{lemma}\label{lem:sumbounds}
The following hold for all $i^*-p \le i < i^*$:
\begin{enumerate} \compactify
\renewcommand{\theenumi}{(\alph{enumi})}
\item\label{it:lowterms}
$4 \veps \cdot B_i 
\ge \sum_{i'\in I\setminus A:\ i'\equiv_p i , i'<i} B_{i'}$
\item\label{it:highterms}
$4 \veps (1+\veps) \cdot B_{i+p} 
\ge \sum_{i'\in I\setminus A:\ i'\equiv_p i , i'>i} B_{i'}$
\item\label{it:allterms}
$4 \veps (1+\veps) \cdot (B_i + B_{i+p}) 
\ge \sum_{i'\in I\setminus A:\ i'\equiv_p i} B_{i'}$
\end{enumerate}
\end{lemma}

\begin{proof}
First recall that we fixed $\delta = (1+\veps)^{-p-1}$, 
and observe that for $i^*-p \le j < i^* + p$
\[
  \delta 
  < (1+\veps)^{-p}
  < (1+\veps)^{i^*-j}
  \le \frac{\|x-y\|}{(1+\veps)^j} 
  \le (1+\veps)^{i^*+1-j}
  \le (1+\veps)^{p+1}
  = \tfrac1\delta,
\]
hence we can apply Theorem \ref{thm:OptDistortion}\ref{it:bilipOpt2} to obtain
\begin{equation}  \label{eq:5}
  \tfrac{1}{1+\veps}
  \le \tfrac{\|\varphi_j(x)-\varphi_j(y)\|}{G_{(1+\veps)^j}(\|x-y\|)} 
  \le 1.
\end{equation}

\begin{enumerate}
\item[\ref{it:lowterms}] 
We first give a lower bound for $B_i$.
Equation \eqref{eq:5} implies that
$\| \varphi_i(x)-\varphi_i(y) \| \ge \frac{G_{(1+\veps)^i}(\|x-y\|)}{(1+\veps)}$.
Lemma \ref{lem:boundG}(i) states that
$G(t)$ is monotone increasing in $t$, 
and noting that
$G(1) > \frac{1+\veps}{2}$ and $i<i^*$, we have
\begin{align*}
B_i
&= \frac{\| \varphi_i(x)-\varphi_i(y) \|}{(1+\veps)^{i(1-\alpha)}} \\
&\ge \frac{G_{(1+\veps)^i}(\|x-y\|)}{(1+\veps)^{i(1-\alpha) + 1}} \\
&= (1+\veps)^{i\alpha - 1} G(\|x-y\|/(1+\veps)^i)    \\
&> (1+\veps)^{i\alpha - 1} G(1)    \\
&\ge \frac{(1+\veps)^{i\alpha}}{2}.
\end{align*}

We now give an upper bound on
$\sum_{i'\in I\setminus A:\ i'\equiv_p i , i'<i} B_{i'}$
when $i<i^*$. 
By Theorem \ref{thm:OptDistortion}\ref{it:lipOpt} and \ref{it:thresholdOpt},
for all $i'\in I$,
\[
  \|\varphi_{i'}(x)-\varphi_{i'}(y)\| \le \min \{\|x-y\| ,(1+\veps)^{i'} \}.
\]
and thus
\begin{align*}
\sum_{i'\in I\setminus A:\ i'\equiv_p i , i'<i} B_{i'}
&= \sum_{i'\in I\setminus A:\ i'\equiv_p i , i'<i} \frac{\| \varphi_i(x)-\varphi(y)
\|}{(1+\veps)^{i'(1-\alpha)}}	\\
&\le \sum_{i'\in I\setminus A:\ i'\equiv_p i , i'<i} 
\frac{(1+\veps)^{i'}}{(1+\veps)^{i'(1-\alpha)}}	\\
&= \sum_{i'\in I\setminus A:\ i'\equiv_p i , i'<i} (1+\veps)^{i'\alpha}.
\end{align*}
Noting that $\tilde \alpha \le \alpha$, and recalling that a
geometric series with ratio less than 1/2 sums to less than twice the
largest term, we have that
\[
 \sum_{i'\in I\setminus A:\ i'\equiv_p i , i'<i} B_{i'}
 \le 2(1+\veps)^{(i-p)\alpha}	
 \le 2\veps (1+\veps)^{i\alpha}
 < 4 \veps B_i.
\]

\item[\ref{it:highterms}]
We first give a lower bound for $B_{i+p}$ when $i^*-p \le i < i^*$. 
Lemma \ref{lem:boundG}(i) states that $G_r(t)$ is monotone 
increasing in $r$.
\begin{align*}
B_{i+p}
& = \frac{\| \varphi_{i+p}(x)-\varphi_{i+p}(y) \|}{(1+\veps)^{(i+p)(1-\alpha)}} \\
& \ge \frac{G_{(1+\veps)^{i+p}}(\|x-y\|)}{(1+\veps)^{(i+p)(1-\alpha) +1}}  \\
& \ge \frac{G_{(1+\veps)^{i^*}}(\|x-y\|)}{(1+\veps)^{(i+p)(1-\alpha) +1 }}    \\
& \ge \frac{(1+\veps)^{i^*} G(1)}{(1+\veps)^{(i+p)(1-\alpha) +1}}     \\
& > \frac{(1+\veps)^{i^*-(i+p)(1-\alpha)}}{2}
\end{align*}

We now give an upper bound on
$\sum_{i'\in I\setminus A:\ i'\equiv_p i , i'>i} B_{i'}$.
By Theorem \ref{thm:OptDistortion}\ref{it:lipOpt} and \ref{it:thresholdOpt},
for all $i'\in I$,  
\[
  \|\varphi_{i'}(x)-\varphi_{i'}(y)\| \le \min \{ \|x-y\| ,(1+\veps)^{i'} \}.
\]
and thus
\begin{align*}
\sum_{i'\in I\setminus A:\ i'\equiv_p i , i'>i} B_{i'}
&= \sum_{i'\in I\setminus A:\ i'\equiv_p i , i'>i}	 
\frac{\| \varphi_i(x)-\varphi(y) \|}{(1+\veps)^{i'(1-\alpha)}}  \\
&\le \sum_{i'\in I\setminus A:\ i'\equiv_p i , i'>i} 
\frac{(1+\veps)^{i^*+1}}{(1+\veps)^{i'(1-\alpha)}} \\
&\le \sum_{i'\in I\setminus A:\ i'\equiv_p i , i'>i} (1+\veps)^{i^*-i(1-\alpha)+1}.
\end{align*}
Noting that $\tilde \alpha \le 1 - \alpha$, and recalling that a
geometric series with ratio less than 1/2 sums to less than twice the
largest term, we have that
\[
 \sum_{i'\in I\setminus A:\ i'\equiv_p i , i'>i} B_{i'}
 \le 2(1+\veps)^{i^* -(i+2p)(1-\alpha) + 1}
 \le 2\veps (1+\veps)^{i^* -(i+p)(1-\alpha) +1}
 < 4 \veps (1+\veps) B_{i+p}.
\]

\item[\ref{it:allterms}] This follows trivially from parts \ref{it:lowterms} and \ref{it:highterms}.
\end{enumerate}
\end{proof}

Now, plugging \eqref{eq:5} and Lemma \ref{lem:sumbounds} into Lemma 
\ref{lem:buckets}, we obtain 
\begin{align*}
  \|\Phi(x)-\Phi(y)\|^2 
  & \ge \sum_{i\in A:\ i < i^*}
	\Big( B_i^2 + B_{i+p}^2 - 2(B_i+B_{i+p})
	\sum_{i'\in I \setminus A:\ i'\equiv_p i} B_{i'} 
	\Big) \\
  &\ge \sum_{i\in A:\ i < i^*}
	\Big( B_i^2 + B_{i+p}^2 - 8\veps(1+\veps)(B_i+B_{i+p})^2 \Big)	\\
  &\ge  (1- 16\veps(1+\veps))
	\sum_{i\in A:\ i < i^*}
	\Big( B_i^2 + B_{i+p}^2 \Big)	\\
  & = 	(1- 16\veps(1+\veps))
	\sum_{i\in A} B_i^2	\\
  &\ge 	\frac{1- 16\veps(1+\veps)}{(1+\veps)^2}
	\sum_{i\in A}
      	\Big(
	\frac{G_{(1+\veps)^i}(\|x-y\|)}{(1+\veps)^{i(1-\alpha)}} 
	\Big)^2 \\
  &\ge 	\frac{1- 16\veps(1+\veps)}{(1+\veps)^2}
	\sum_{i^*-p \le i < i^*+p}
      	\Big(
	(1+\veps)^{i\alpha} \cdot G((1+\veps)^{i^*-i}) 
	\Big)^2 \\
  &\ge  (1- 16\veps(1+\veps))(1+\veps)^{2i^*\alpha-2}
	\sum_{b:\ -p \le b < p}
      \Big( (1+\veps)^{b\alpha}\cdot G((1+\veps)^{-b}) \Big)^2 \\
  &>  (1- 16\veps(1+\veps))(1+\veps)^{-4} \|x-y\|^{2\alpha}
	\sum_{b:\ -p \le b < p}
      \Big( (1+\veps)^{b\alpha}\cdot G((1+\veps)^{-b}) \Big)^2.
\end{align*}
and similarly, using also Lemma \ref{lem:boundG} and recalling that 
$\veps < \frac{1}{4}$ 
\begin{align*}
  \|\Phi(x)-\Phi(y)\|^2 
  &\le \sum_{i\in A:\ i < i^*}
	\Big( B_i^2 + B_{i+p}^2 + 2(B_i+B_{i+p}) 
	\sum_{i'\in I\setminus A:\ i'\equiv_p i} B_{i'} +
	\Big(
	\sum_{i'\in I\setminus A:\ i'\equiv_p i} B_{i'} 
	\Big)^2	
	\Big) \\
  &< (1+ 24\veps(1+\veps)^2)
	\sum_{i\in A:\ i < i^*}
	\Big( B_i^2 + B_{i+p}^2 \Big)	\\
  &= 	(1+ 24\veps(1+\veps)^2)
	\sum_{i\in A} B_i^2	\\
  &\le 	(1+ 24 \veps(1+\veps)^2)
	\sum_{i\in A}
      	\Big(
	\frac{G_{(1+\veps)^i}(\|x-y\|)}{(1+\veps)^{i(1-\alpha)}} 
	\Big)^2 \\
  &\le 	(1+ 24 \veps(1+\veps)^2)
	\sum_{i^*-p \le i < i^*+p}
      	\Big(
	(1+\veps)^{i \alpha} \cdot G((1+\veps)^{i^*-i+1}) 
	\Big)^2 \\
  &\le  (1+ 24 \veps(1+\veps)^2) (1+3\veps)^2 (1+\veps)^{2i^*\alpha}
	\sum_{b:\ -p \le b < p}
	\Big( 
	(1+\veps)^{b \alpha}\cdot G((1+\veps)^{-b}) 
	\Big)^2 \\
  &\le  (1+ 24 \veps(1+\veps)^2)(1+3\veps)^2 \|x-y\|^{2\alpha}
	\sum_{b:\ -p \le b < p}
      \Big( (1+\veps)^{b\alpha}\cdot G((1+\veps)^{-b}) \Big)^2.  
\end{align*}

Setting 
$M=\sum_{b:\ -p \le b < p}\Big( (1+\veps)^{b\alpha}\cdot 
G((1+\veps)^{-b}) \Big)^2$,
which clearly depends only on $\veps,\alpha$ 
(and is in particular independent of $x,y$),
we combine the last two estimates to obtain
\[
  \frac{1- 16\veps(1+\veps)}{(1+\veps)^4}
  < \frac{\|\Phi(x)-\Phi(y)\|^2}{M\ \|x-y\|^{2\alpha}}
  \le (1+ 24\veps(1+\veps)^2)(1+3\veps)^2
\]
We conclude that the final embedding $\Phi/\sqrt M$ achieves distortion $1+O(\veps)$
to $\|x-y\|^{\alpha}$ for all $0< \alpha < 1$, and this concludes the proof of
Theorem \ref{thm:main}.

\section{Extension to {\boldmath $\ell_1$} Space}
\label{sec:lp}

We explain how our results and techniques can be extended 
to $\ell_1$.
A number of key tools used in our previous embeddings are specific 
to $\ell_2$, for example the JL-Lemma, the Gaussian transform,
and the Kirszbraun theorem,
and we must therefore find suitable replacements for these tools.
Note however that there is no Lipschitz extension theorem for $\ell_1$. 

The primary result of this section is
a variant of our snowflake embedding, Theorem \ref{thm:main}.\footnote{Subsequent
to the publication of this result in Proceedings of SODA 2011,
Bartal and Gottlieb \cite{BG14} presented a new single-scale embedding for all
$\ell_p$, $1 \le p < 2$, and derived a snowflake embedding for
$\ell_p$ with only polynomial dependence on the doubling dimension.}
We note that the snowflake operator is necessary in this theorem,
as for $\alpha=1$ 
Lee, Mendel and Naor \cite[Theorem 1.3]{LMN05} have shown 
that the target dimension 
cannot be bounded as a function of $\lambda(S)$, independently of $|S|$.

\begin{theorem} \label{thm:mainlp}
Let $0<\veps<1/4$, $0<\alpha<1$ and $p\in\{1,\infty\}$.
Every finite subset $S\subset \ell_1$ with $\lambda=\lambda(S)$ 
admits an embedding $\Phi:S\to \ell_1^k$ satisfying
$$
   1
   \le \frac{\|\Phi(x)-\Phi(y)\|_1}{\|x-y\|_1^{\alpha}}
   \le 1+\veps,
   \qquad \forall x,y\in S;
$$
with
$k  = \mathrm{exp}\{\lambda^{O(\tilde{\alpha}\log(1/\veps)+\log\log\lambda)}\}$.
\end{theorem}

Recall that our (refined) single scale embedding for $\ell_2$ (Theorem 
\ref{thm:OptDistortion}), coupled with an application of Assouad's 
technique, were sufficient to prove Theorem \ref{thm:main}. 
Similarly, a single scale embeddings for $\ell_1$, coupled with a 
standard application of Assouad's technique, is sufficient to prove 
Theorem \ref{thm:mainlp}. We present a single scale embedding for $\ell_1$ 
below, and Theorem \ref{thm:mainlp} then follows easily.


\subsection{Single Scale Embedding for $\ell_1$}

We can extend Theorem \ref{thm:OptDistortion} to $\ell_1$ spaces as follows.
For $r>0$ define $L_r:\R\to\R$, called the Laplace distance transform, by 
$L_r(t)=r(1-e^{-t/r})$. Observe that $L_r(t)=r\cdot G(\sqrt{t/r})^2$.

\begin{theorem} \label{thm:L1SingleScale}
For every scale $r>0$ and every $0<\delta,\veps<1/4$,
every finite set $S\subset \ell_1$ admits an embedding 
$\varphi:S\to \ell_1^k$ for 
$k=\mathrm{exp}\{\lambda^{O(\log(1/\veps\delta)+\log\log\lambda)}\}$,
satisfying:
\begin{enumerate} \compactify
\renewcommand{\theenumi}{(\alph{enumi})}
\item\label{it:lipOptL1}
\lip condition:
$\|\varphi(x)-\varphi(y)\|_1 \le \|x-y\|_1$
for all $x,y\in S$.
\item\label{it:bilipOpt2L1}
$1+\veps$ distortion to the Laplace transform (at scales near $r$):
  For all $x,y\in S$ with $\delta r \le \|x-y\|_1\le \frac{r}{\delta}$,
$$
  \frac{1}{1+\veps}
  \le \frac{\|\varphi(x)-\varphi(y)\|_1}{L_r(\|x-y\|_1)} \le 1.
$$
\item\label{it:thresholdOptL1}
Boundedness: $\|\varphi(x)\|_1 \leq r$ for all $x\in S$.
\end{enumerate}
\end{theorem}

\begin{proof}[Proof Sketch]
We would like to utilize the framework designed for $\ell_2$ 
in Section \ref{sec:l2BetterSingleScale}.
However, a few problems arise. 
Let us point them out explain how to solve them.

\begin{itemize} \compactify
\item Step 7: 
This step is not possible for $\ell_1$ norm,
since there is no $\ell_1$-analogue of the Kirszbraun theorem. 
Instead, we modify the entire construction (specifically, steps 2-6) 
so that they work with the entire data set $S$, not only with the net $N$.
The effect of this will be seen shortly.
(In $\ell_2$, the same approach of discarding step 7 can be achieved 
by applying the Kirszbraun Theorem separately in every cluster in step 4,
but this approach does not seem to have any advantages.)
\item Step 2: 
We apply a padded decomposition to the entire set $S$ 
(and not only to the net $N$) with essentially the same parameters and bounds.
Thus, from now on each cluster $C\in P_i$ is a subset of $S$ 
(rather than of $N$).
\item Step 3: 
Instead of the Gaussian transform, we apply the Laplace transform
$L_r$, i.e. $g_C$ now satisfies 
$\|g_C(x)-g_C(y)\|_1  = L_r(\|x-y\|)$ for all $x,y\in C$.
Such an embedding $g_C:C\to l_1$ is known to exist, 
see \cite[Corollary 9.1.3]{DL97}.
The effect is clearly quite similar to that of the Gaussian transform.
The fact that $C$ is not a subset of $N$ is not an issue.
\item Step 4: 
We need to find a weak analogue to the JL lemma, but there is 
an additional complication of having to deal with points not in 
the net $N$.
Specificially, we need a map $\Psi:g_C(C)\to \ell_1^{k'}$ 
which satisfies: 
(i) $\Psi$ is $1$-\lip on the entire cluster $g_C(C)$; and 
(ii) $\Psi$ achieves $1+\veps$ distortion
on the cluster net points $g_C(C\cap N)$.
Observe that the former requirement is non-standard and does not follow
from ``standard'' dimension reduction theorems for finite subsets of 
$\ell_1$.
We will utilize the simple $\ell_1$ dimension-reduction mapping embedding 
described below in Theorem \ref{thm:l1-isometry} (recall $g_C(C) \subset \ell_1$), 
that achieves dimension $k'=2^{|C\cap N|}$ where 
$|C\cap N| \le \lambda^{O(\Delta/\veps\delta r)}$.\footnote{We suspect 
that the dimension can be further reduced, since
the construction of Theorem \ref{thm:l1-isometry} is an isometry on 
$g_C(C\cap N)$, and does not exploit the $1+\veps$ distortion allowed by 
requirement (ii). However, an improved map $\Psi$ cannot be linear, since in the worst case 
such a linear map requires dimension 
$k = 2^{\Omega(|C\cap N|)}$ \cite[Corollary 12.A]{FJS91}.}

\item Step 5: 
There is only a minor change;
since we do not restrict attention to net points,
we now define $h_C(x)=\min_{y\in S\setminus C} \|x-y\|_1$.
\item Step 6: 
There is only a minor change to the scaling factor,
namely $\varphi = m^{-1} \bigoplus^m_{i=1} \varphi_i$.
\end{itemize}

The rest of the proof is quite similar to the one presented for $\ell_2$,
and the final dimension obtained is 
$mk'
  = O(\veps^{-1}\log\lambda\log\log\lambda)\cdot \mathrm{exp}(\lambda^{O(\log(\veps^{-1}\delta^{-2}\log\lambda))})
  = \mathrm{exp}\{\lambda^{O(\log(1/\veps\delta)+\log\log\lambda)}\}$.
\end{proof}

It remains only to present the following construction,
which was observed jointly with Gideon Schechtman.

\begin{theorem}\label{thm:l1-isometry}
Given a point set $\tilde{C} \subset \ell_1$ and a subset $D \subset \tilde{C}$, 
there exists a map
$\Psi:\tilde{C} \to \ell_1^{k}$ 
with $k = 2^{|D|}$
which satisfies: 
(i) $\Psi$ is $1$-\lip on all of $\tilde{C}$; and 
(ii) $\Psi$ is an isometry on the subset $D$.
\end{theorem}

\begin{proof}
Construct $\Psi$ as follows. Since the metric $\tilde{C} \subset l_1$, 
it can be written as a conic combination of cut metrics, 
i.e.\ there are $\gamma_A\ge 0$ for $A\subset \tilde{C}$ such that 
\[
  \|x-y\|_1 = \sum_A \gamma_A |1_A(x)-1_A(y)|,\qquad \forall x,y\in \tilde{C},
\]
where $1_A(x)=1$ if $x\in A$ and $0$ otherwise. In other words, 
$x\to\sum_A \gamma_A 1_A(x)$ is an isometric embedding of $\tilde{C}$ 
into $\ell_1$.
Let $\Psi$ have one coordinate for every subset $B\subset D$;
this coordinate is given by $x\to \sum_{A:\ A\cap D =B} \gamma_A 
1_A(x)$. In words, we add together coordinates that correspond 
to different $A$ but have the same $A \cap D$.
Observe that $\Psi$ is $1$-\lip for all $x\in \tilde{C}$,
simply because adding two coordinates together can only decrease distances,
and that this $\Psi$ is an isometry on $D$, 
because for all $x,y \in D$ if coordinates corresponding to $A$ and $A'$ are added together
then necessarily $1_A(x)=1_{A'}(x)$ and similarly $1_A(y)=1_{A'}(y)$.
Observe that $k=2^{|D|}$.
\end{proof}


\section{Algorithmic applications}\label{sec:app}

Here we illustrate the effectiveness and potential of our results 
for various algorithmic tasks by describing two immediate 
(theoretical) applications.

\paragraph{Distance Labeling Scheme (DLS).}
Consider this problem for the family
of $n$-point $\ell_2$ metrics with a given bound on the doubling dimension.
As usual, we assume the interpoint distances are in the range $[1,R]$.
Our snowflake embedding into $\ell_2^k$ (Theorem \ref{thm:main} for 
$\alpha = \frac{1}{2}$)
immediately provides a DLS with approximation $(1+\veps)^2\leq 1+3\veps$,
simply by rounding each coordinate to a multiple of $\veps/2k$. We have:

\begin{lemma}
Every finite subset $\ell_2$ with $\lambda = \lambda(S)$ possesses a 
$(1+\veps)$-approximate distance labeling scheme with label size
$$
  k\cdot \log \tfrac{R}{\veps/2k}
  = \tilde O(\veps^{-4}(\dimC S)^2)\log R.
$$
\end{lemma}

Notice that, apart from the $\log R$ term, this bound is independent of $n$.
The published bounds of this form (see \cite{HM06} and references therein)
apply to the the more general family of all doubling metrics
(not necessarily Euclidean) but require exponentially larger label size,
roughly $(1/\veps)^{O(\dimC S)}$.

\paragraph{Approximation algorithms for clustering.}

Clustering problems are often defined as an optimization problem whose 
objective function is expressed in terms of distances between data points. 
For example, in the $k$-center problem one is given a metric $(S,d)$ and 
is asked to identify a subset of centers $C \subset S$ that minimizes the 
objective $\max_{x \in S}d(x,C)$. When the data set $S$ is Euclidean (and 
the centers are discrete, i.e.\ from $S$), one can apply our snowflake 
embedding (Theorem \ref{thm:main}) and solve the problem in the target 
space, which has low dimension $k$. Indeed, it is easy to see how to map 
solutions from the original space to the target space and vice versa, with 
a loss of at most a $(1+\veps)^2 \le 1+3\veps$ factor in the objective.

For other clustering problems, like $k$-median or min-sum clustering,
the objective function is the \emph{sum} of certain distances.
The argument above applies, except that now in the target space
we need an algorithm that solves the problem with $\ell_2$-squared costs.
For instance, to solve the $k$-median problem in the original space,
we can use an algorithm for $k$-means in the target space.
Schulman \cite{Schulman00} has designed algorithms 
for min-sum clustering under both $\ell_2$ and $\ell_2$-squared costs,
and their run time depend exponentially on the dimension. The following 
lemma follows from our snowflake embedding and \cite[Propositions 
14,28]{Schulman00}. For simplicity, we will assume that $k=O(1)$. 

\begin{lemma}
Given a set of $n$ points $S \in \R^d$, a $(1+\veps)$-approximation to 
the $\ell_2$ min-sum $k$-clustering for $S$, for $k=O(1)$, can be computed 
\begin{enumerate}
\item
in deterministic time
$n^{O(d')} 2^{2^{(O(d'))}} $.
\item 
in randomized time
$n^{O(1)} + n'^{O(d')} 2^{2^{(O(d'))}} $,
where $n' = O(
\veps^{-2} 
\log (\delta^{-1} n)
)$,
with probability $1-\delta$.
\end{enumerate}
where $d' = \min\{d,\tilde{O}(\veps^{-4}\dim^2 S)\}$.
\end{lemma}

\subsubsection*{Acknowledgments}
The authors thank Assaf Naor and Gideon Schechtman for useful discussions
and references, and Yair Bartal for helpful comments on an earlier
version of this paper.

{\small 
\bibliographystyle{alpha}
\bibliography{robi,drafts}
}

\appendix

\section{Omitted proof.}\label{sec:A}

\begin{proof}[Proof of Lemma \ref{lem:boundG}]
For assertion (i), since $e^{-t^2/r^2}$ is decreasing in $t$,
$G_r(t) = r(1-e^{-t^2/r^2})^{1/2}$ is increasing in $t$.
Now consider the function $F(x) = \frac{1-e^{-x}}{x}$.
This function is monotone decreasing in $x>0$: The 
derivative of $F(x)$ is $\frac{(x+1)e^{-x} - 1}{x^2}$, 
and it is easily verified that the numerator is negative 
whenever $x>0$. It follows that 
$G_r(t) = t \sqrt{F(t^2/r^2)}$ is monotone increasing in
$r$, completing assertion (i). Further
$\frac{G_r(t)}{t} = \sqrt{F(t^2/r^2)}$ is monotone 
decreasing in $t$, which proves assertion (ii).

For assertion (iii), recall from assertion (i) 
that $G_r(t)$ is monotonically increasing (in $t$), and thus
\[
  \frac{G_r(t')}{G_r(t)}
  \le \frac{G_r((1+\eta)t)}{G_r(t)}
  \le \frac{G((1+\eta)t/r)}{G(t/r)}.
\]
Letting $s=t/r$, we have
\begin{equation}  \label{eq:3}
  \frac{G((1+\eta)s)^2}{G(s)^2} - 1
  = \frac{G((1+\eta)s)^2-G(s)^2}{G(s)^2} 
  = \frac{e^{-s^2} - e^{-(1+\eta)^2 s^2}}{1-e^{-s^2}}
  \le \frac{e^{-s^2}(1-e^{-3\eta s^2})}{1-e^{-s^2}}.
\end{equation}
Recall that by the Taylor series expansion, 
$e^{-z} = 1 - z + \frac{z^2}{2} - \frac{z^3}{6} + \ldots$,
and so for all $0\le z\le 1$ we have 
$ 1-z \le e^{-z} \le 1-z+z^2/2 \le 1-z/2$.
Using this estimate, we now have three cases:
\begin{itemize} \compactify
\item 
When $s^2\le 1$, the righthand side of \eqref{eq:3} is at most
$\frac{1\cdot 3\eta s^2}{s^2/2} \le 6\eta$.
\item
When $1\le s^2\le 1/3\eta$, the righthand side of \eqref{eq:3} is at most
$\frac{e^{-s^2}\cdot 3\eta s^2}{1-1/e} \le 6\eta s^2 e^{-s^2}\le 6\eta/e$,
where the last inequality follows from the observation that
$z\mapsto ze^{-z}$ is monotonically decreasing for all $z\ge 1$.
\item
When $s^2\ge 1/3\eta$, the righthand side of \eqref{eq:3} is at most
$\frac{e^{-s^2}\cdot 1}{1-1/e} \le \frac{e^{-s^2}\cdot 3\eta s^2}{1-1/e} \le 6\eta/e$,
where the last inequality follows similarly to the previous case.
\end{itemize}
Altogether, we conclude that
$
  \frac{G_r(t')}{G_r(t)}
  \le \frac{G((1+\eta)s)}{G(s)}
  \le \sqrt{1+6\eta}
  < 1+3\eta.
$
\end{proof}

\remove{
\section{Basic properties of Lipschitz functions}\label{sec:lip-prop}

\begin{enumerate} \compactify
\renewcommand{\theenumi}{(P\arabic{enumi})}
\item \label{it:lipscalar}
Let $f:X\to \ell_2^k$ and $\alpha>0$. Then
$\lipnorm{\alpha f} \leq \alpha\lipnorm{f}$.
\item \label{it:lipsum}
Let $f_1,\ldots,f_m:X\to \ell_2^k$. Then the sum
$\sum_{i=1}^m f_i$ which maps $x\to f_1(x)+\ldots+f_m(x)\in\ell_2^k$
has \lip norm
$\lipnorm{\sum_{i=1}^m f_i}
  \leq (\sum_{i=1}^m \lipnorm{f_i})^{1/2}$.

Similarly, the direct sum $\bigoplus_{i=1}^m f_i$
which maps $x\to f_1(x)\oplus\cdots\oplus f_m(x)\in\ell_2^{mk}$,
has \lip norm
$\lipnorm{\bigoplus_{i=1}^m f_i} \leq (\sum_{i=1}^m \lipnorm{f_i})^{1/2}$.
Both cases can be further bounded by $m^{1/2} \max_{i=1,\ldots,m} \lipnorm{f_i}$.
\item \label{it:lipcompose}
Let $f:X\to Y$ and $g:Y\to Z$.
Then their composition $g\circ f$ mapping $x\to g(f(x))\in Z$
has \lip norm
$\lipnorm{g\circ f}\leq \lipnorm{f}\cdot \lipnorm{g}$.
\item \label{it:lipproduct}
Let $f:X\to \ell_2^k$ and $g:X\to\R$.
Then their product $fg:x\to g(x)\cdot f(x)$ has \lip norm
\[
  \lipnorm{fg} \leq
  \lipnorm{f}\cdot \max_x |g(x)| + \lipnorm{g}\cdot \max_x \|f(x)\|.
\]
\item \label{it:lipthreshold}
Let $f:X\to \R$ and $T>0$. Then thresholding $f$ at value $T$,
i.e. $g:x\to \minn{f(x),T}$,
has \lip norm $\lipnorm{g} \le \lipnorm{f}$.
\end{enumerate}
}

\end{document}